\documentclass[11pt]{article}
\usepackage{hyperref}
\usepackage{times}  
\usepackage{mathpazo}
\usepackage{amssymb,amsmath,amsthm}
\usepackage{epsfig}

\newtheorem{theorem}{Theorem}[section]
\newtheorem{proposition}[theorem]{Proposition}
\newtheorem{definition}[theorem]{Definition}

\newtheorem{lemma}[theorem]{Lemma}

\newtheorem{corollary}[theorem]{Corollary}
\newtheorem{fact}[theorem]{Fact}

\newcommand{\qedsymb}{\hfill{\rule{2mm}{2mm}}}
\renewenvironment{proof}[1][]{\begin{trivlist}
\item[\hspace{\labelsep}{\bf\noindent Proof#1:\/}] }{\qedsymb\end{trivlist}}

\def\Z{{\mathbb{Z}}}
\def\R{\mathbb{R}}

\def\mod{\mbox{mod}}

\newcommand{\NP}{\mathsf{NP}}

\newcommand{\eps}{\epsilon}
\renewcommand{\epsilon}{\varepsilon}

\newcommand{\rank}{\mathop{\mathrm{rank}}}
\newcommand{\minrank}{\mathop{\mathrm{minrk}}}

\newcommand{\Fset}{\mathbb{F}}         

\newcommand{\vchrom}{{\chi_v}}
\newcommand{\svchrom}{{\chi^{(s)}_v}}

\begin{document}

\title{{\bf On Minrank and the Lov\'asz Theta Function}}

\author{
Ishay Haviv\thanks{School of Computer Science, The Academic College of Tel Aviv-Yaffo, Tel Aviv 61083, Israel.
}
}

\date{}

\maketitle

\begin{abstract}
Two classical upper bounds on the Shannon capacity of graphs are the $\vartheta$-function due to Lov\'asz and the minrank parameter due to Haemers.
We provide several explicit constructions of $n$-vertex graphs with a constant $\vartheta$-function and minrank at least $n^\delta$ for a constant $\delta>0$ (over various prime order fields).
This implies a limitation on the $\vartheta$-function-based algorithmic approach to approximating the minrank parameter of graphs.
The proofs involve linear spaces of multivariate polynomials and the method of higher incidence matrices.
\end{abstract}

\section{Introduction}\label{sec:intro}

For a graph $G$ on the vertex set $V$, let $G^k$ denote the graph on the vertex set $V^k$ in which two distinct vertices $(u_1,\ldots,u_k)$ and $(v_1,\ldots,v_k)$ are adjacent if for every $1 \leq i \leq k$ it holds that $u_i$ and $v_i$ are either equal or adjacent in $G$. The {\em Shannon capacity} of $G$, introduced by Shannon in 1956~\cite{Shannon56}, is defined as the limit $c(G) = \lim_{k \rightarrow \infty}{\sqrt[k]{\alpha(G^k)}}$, where $\alpha(G^k)$ stands for the independence number of $G^k$.
The study of the graph parameter $c(G)$ is motivated by an application in information theory, as it measures the effective alphabet size in a communication over a noisy channel represented by $G$. However, computing the Shannon capacity of a graph is a notoriously difficult task. Its exact value is not known even for small graphs, such as the cycle on $7$ vertices, and from a computational perspective, it is not known if the problem of deciding whether the Shannon capacity of a given graph exceeds a given value is decidable.

The difficulty in computing the Shannon capacity of graphs motivates studying upper and lower bounds on $c(G)$. It is known that $c(G)$ is sandwiched between the independence number $\alpha(G)$ of $G$ and its clique cover number $\chi(\overline{G})$. In 1979, Lov{\'a}sz~\cite{Lovasz79} introduced the $\vartheta$-function of graphs defined as follows: For a graph $G$ on the vertex set $V$, $\vartheta(G)$ is the minimum of $\max_{i \in V} \frac{1}{\langle x_i, y\rangle ^2}$, taken over all choices of unit vectors $y$ and $(x_i)_{i \in V}$ such that $x_i$ and $x_j$ are orthogonal whenever $i$ and $j$ are distinct non-adjacent vertices in $G$ (see~\cite{Knuth94} for several equivalent definitions).
It was shown in~\cite{Lovasz79} that $c(G) \leq \vartheta (G)$ for every graph $G$, and this was used to prove that the Shannon capacity of the cycle on $5$ vertices is equal to $\sqrt{5}$.
The $\vartheta$-function of graphs can be computed in polynomial running time at an arbitrary precision using semi-definite programming~\cite{GrotschelLS81} and it has found interesting combinatorial and algorithmic applications over the years (see, e.g.,~\cite{Feige97,AlonK98,OfekF06}).

Another upper bound on the Shannon capacity of graphs is the minrank parameter introduced by Haemers~\cite{Haemers79,Haemers81}.
For a graph $G$ on the vertex set $V = \{1,\ldots,n\}$, the minrank of $G$ over a field $\Fset$, denoted ${\minrank}_\Fset(G)$, is the minimum of ${\rank}_\Fset(M)$ over all matrices $M \in  \Fset^{n \times n}$ satisfying $M_{i,i} \neq 0$ for every $i \in V$, and $M_{i,j} = 0$ whenever $i$ and $j$ are distinct non-adjacent vertices in $G$. For the field $\Fset_p$ of prime order $p$ we use the notation ${\minrank}_p(G)$.
For most graphs the minrank parameter is larger than the $\vartheta$-function~\cite{Coja-Oghlan05,HavivL12,Golovnev0W17}, yet it was shown in~\cite{Haemers79} that there are graphs for which the minrank bound on the Shannon capacity is tighter. In recent years, the minrank parameter has attracted an intensive research motivated by its relations to various topics in information theory, e.g., (linear) index coding~\cite{BirkKol98,BBJK06,LS07,BlasiakKL13}, network coding~\cite{AhlswedeCLY00,ESG08}, distributed storage~\cite{mazumdar2014duality,ji2014caching}, and wireless communication~\cite{maleki2014index,jafar2014topological}, and in theoretical computer science, e.g., Valiant's approach to lower bounds in circuit complexity~\cite{Valiant92,Riis07,Golovnev0W17}, communication complexity~\cite{PudlakRS97}, and randomized computation~\cite{HavivL13}.

The computational problem of deciding whether the minrank of a given graph is at most $3$ is known to be $\NP$-complete over any fixed finite field~\cite{Peeters96}.
Moreover, assuming a certain variant of the unique games conjecture, it is $\NP$-hard to approximate the minrank of a given graph to within any constant~\cite{LangbergS08} (and even to within certain super-constant factors, as follows from~\cite{LangbergS08} combined with~\cite{DinurS10}).
On the algorithmic side, relations between the minrank parameter and the tractable $\vartheta$-function can be beneficial to efficient approximation algorithms for minrank. This approach was taken in~\cite{ChlamtacH14} where it was proved that every graph $G$ with ${\minrank}_2(G)=k$ satisfies $\vartheta(G) \leq 2^{k/2}+1-2^{1-k/2}$. This bound was used to obtain an efficient algorithm that given an $n$-vertex graph $G$ with ${\minrank}_2(G)=k$, where $k$ is a constant, finds a clique cover of $G$ of size $O(n^{\alpha(k)})$ for some $\alpha(k)<1$ (e.g., $\alpha(3) \approx 0.2574$). Note that such a clique cover of $G$ in particular yields a matrix confirming the same bound on the minrank.

The algorithm of~\cite{ChlamtacH14} for minrank employs the semi-definite programming technique used in the algorithm of Karger, Motwani, and Sudan for graph coloring~\cite{KargerMS98}. The analysis of the latter shows that every $n$-vertex graph $G$ with a constant $\vartheta(G)$ has a clique cover of size $O(n^\alpha)$ for a constant $\alpha<1$.
This is known to be tight in the sense that there are $n$-vertex graphs $G$ with a constant $\vartheta(G)$ and yet a clique cover number $n^{\Omega(1)}$~\cite{KargerMS98,Charikar02} (see also~\cite{FeigeLS04}). However, the minrank of a graph might in general be much smaller than its clique cover number (even exponentially; see~\cite{Haemers81}).
It is natural to ask, then, whether a constant $\vartheta(G)$ guarantees a stronger bound of $n^{o(1)}$ on ${\minrank_2}(G)$. In the current work we rule out this possibility in a general sense, as stated below.

\begin{theorem}\label{thm:Intro1}
For every prime $p$ there exist $c = c(p)$ and $\delta = \delta(p) >0$ such that for infinitely many integers $n$ there exists an $n$-vertex graph $G$ such that $\vartheta(G) \leq c$ and $\minrank_p(G) \geq n^\delta$.
\end{theorem}
\noindent
Note that for the special case of $p=2$ we obtain an $n$-vertex graph $G$ with $\vartheta(G) \leq 16$ and $\minrank_2(G) \geq n^{0.1499}$. This implies a limitation on the $\vartheta$-function-based algorithmic approach of~\cite{ChlamtacH14} to minrank over $\Fset_2$. 

We also obtain the following result in which the prime $p$ is not a constant.

\begin{theorem}\label{thm:Intro2}
There exist $c$ and $\delta>0$ such that for infinitely many integers $n$ there exists an $n$-vertex graph $G$ such that $\vartheta(G) \leq c$ and $\minrank_p(G) \geq n^{\delta}$ for some prime $p = \Theta(\log n)$.
\end{theorem}

In our final construction, the bound on the minrank holds over any field of a sufficiently large prime order. However, the bound on the $\vartheta$-function is relaxed to a bound on the vector chromatic number $\chi_v$ of the graph's complement (see Definition~\ref{def:chi_v}).

\begin{theorem}\label{thm:Intro3}
There exists a constant $\delta >0$ such that for infinitely many integers $n$ there exists an $n$-vertex graph $G$ such that $\chi_v(\overline{G}) \leq 3$ and $\minrank_p(G) \geq n^\delta$ for any prime $p \geq \Omega( \log n)$.
\end{theorem}

All the aforementioned constructions are explicit and belong to the family of generalized Kneser graphs (see Definition~\ref{def:Kneser}).
Our technical contribution lies in presenting two general methods for proving bounds on the minrank parameter, employing the tools of linear spaces of multivariate polynomials and higher incidence matrices (see, e.g.,~\cite[Chapters~5~and~7]{BabaiF92}).
We demonstrate the usefulness of these tools in studying the minrank of additional graph families (see Sections~\ref{subsec:ortho} and~\ref{subsec:directed}) and expect our techniques to have further applications in the future.

\subsection{Techniques and Related Work}

As mentioned above, the constructions given in Theorems~\ref{thm:Intro1},~\ref{thm:Intro2}, and~\ref{thm:Intro3} are all from the family of generalized Kneser graphs. In these graphs the vertices are all subsets of a given size of some universe, and two distinct vertices are adjacent if their intersection size lies in a certain specified set of sizes. We are particularly interested in those graphs with only one intersection size in the specified set, as the $\vartheta$-function of their complement is easily bounded (see Lemma~\ref{lemma:chi_vK}).

The independence numbers of generalized Kneser graphs correspond to well-studied combinatorial questions on the size of uniform set systems with forbidden intersection sizes (see, e.g.,~\cite{FranklR87}).
Tools from linear algebra are often used in proving upper bounds in such scenarios. This includes the celebrated works of Ray-Chaudhuri and Wilson~\cite{Ray75} and Frankl and Wilson~\cite{FranklW81}, who obtained their bounds using the method of higher incidence matrices (more specifically, inclusion matrices; see~\cite[Chapter~7]{BabaiF92}). Alon, Babai, and Suzuki~\cite{AlonBS91} provided alternative proofs and generalizations using a different approach operating on linear spaces of multivariate polynomials (see~\cite[Chapter~5]{BabaiF92}). These results have found numerous applications in combinatorics and in theoretical computer science, e.g., explicit constructions in Euclidean Ramsey theory~\cite{AlonP91}, counterexamples to Borsuk's conjecture~\cite{KahnK93}, and integrality gap constructions for approximating graph parameters such as the chromatic number~\cite{KargerMS98}, the independence number~\cite{Feige97,AlonK98}, and the vertex cover number~\cite{Charikar02}.

While the above results provide strong upper bounds on the independence numbers of certain generalized Kneser graphs, they do not imply any meaningful bounds on their minrank. Nevertheless, we show in this work that both the tools of higher incidence matrices and multivariate polynomials can be used to obtain upper bounds on the minrank parameter as well.
We demonstrate these techniques and apply them to several graph families (most, but not all, of which are of the Kneser type).
To obtain the lower bounds on the minrank in Theorems~\ref{thm:Intro1},~\ref{thm:Intro2}, and~\ref{thm:Intro3}, we apply a known relation between the minrank of a graph and the minrank of its complement (see Lemma~\ref{lemma:minrank_comp}).

We note that Alon used in~\cite{AlonUnion98} multivariate polynomials to obtain an upper bound, closely related to minrank, on the Shannon capacity of graphs.
In addition, Lubetzky and Stav used inclusion matrices in~\cite{LS07} to prove that for every prime $p$, an $n$-vertex graph can have a multiplicative gap of $n^{0.5-o(1)}$, in either direction, between the $\vartheta$-function and the minrank over $\Fset_p$. (Note that the bound on $\vartheta$ in these constructions is of $\sqrt{n}$.) It will be interesting to figure out if our construction in Theorem~\ref{thm:Intro1} combined with the randomized graph product technique of~\cite{BermanS92,Feige97} can be used to improve on this multiplicative gap.

\paragraph{Outline.}
In Section~\ref{sec:preliminaries} we gather a few needed definitions and lemmas.
In Sections~\ref{sec:poly} and~\ref{sec:incidence} we prove upper bounds on the minrank of several graph families using, respectively, linear spaces of multivariate polynomials and inclusion matrices. Finally, in Section~\ref{sec:separation}, we prove Theorems~\ref{thm:Intro1},~\ref{thm:Intro2}, and~\ref{thm:Intro3}.

\section{Preliminaries}\label{sec:preliminaries}

Unless otherwise specified, a graph will refer to a simple undirected graph.
We use the notation $[d] = \{1,2,\ldots,d\}$.

\subsection{Minrank}

The minrank of a graph over a field $\Fset$ is defined as follows.
\begin{definition}\label{def:minrank}
Let $G=(V,E)$ be a directed graph on the vertex set $V = \{1,\ldots,n\}$ and let $\Fset$ be a field.
We say that an $n$ by $n$ matrix $M$ over $\Fset$ {\em represents} $G$ if $M_{i,i} \neq 0$ for every $i \in V$, and $M_{i,j}=0$ for every distinct $i,j \in V$ such that $(i,j) \notin E$.
The {\em minrank} of $G$ over $\Fset$ is defined as
\[{\minrank}_\Fset(G) =  \min\{{\rank}_{\Fset}(M)\mid M \mbox{ represents }G\mbox{ over }\Fset\}.\]
\end{definition}

The above definition is naturally extended to undirected graphs by replacing every undirected edge with two oppositely directed edges.
Note that for a prime $p$ we write ${\rank}_{p}(M) = {\rank}_{\Fset_p}(M)$ and ${\minrank}_{p}(G) = {\minrank}_{\Fset_p}(G)$.

We need the following lemma that relates the minrank of a graph to the minrank of its complement. For a proof see, e.g.,~\cite[Remark~2.2]{Peeters96},~\cite[Claim~2.5]{LS07}.

\begin{lemma}\label{lemma:minrank_comp}
For every field $\Fset$ and an $n$-vertex graph $G$, $\minrank_\Fset(G) \cdot \minrank_\Fset(\overline{G}) \geq n$.
\end{lemma}

\subsection{Vector Chromatic Number}

Consider the following two relaxations of the chromatic number of a graph, due to Karger, Motwani, and Sudan~\cite{KargerMS98}.

\begin{definition}\label{def:chi_v}
For a graph $G=(V,E)$ the {\em vector chromatic number} of $G$,
denoted $\vchrom(G)$, is the minimal real value of $\kappa > 1$ such
that there exists an assignment of a unit vector $w_i$ to each
vertex $i \in V$ satisfying the inequality $\langle w_i, w_j \rangle \leq
-\frac{1}{\kappa -1}$ whenever $i$ and $j$ are adjacent in $G$.
\end{definition}

\begin{definition}
For a graph $G=(V,E)$ the {\em strict vector chromatic number} of
$G$, denoted $\svchrom(G)$, is the minimal real value of $\kappa > 1$
such that there exists an assignment of a unit vector $w_i$ to each
vertex $i \in V$ satisfying the equality $\langle w_i, w_j \rangle =
-\frac{1}{\kappa -1}$ whenever $i$ and $j$ are adjacent in $G$.
\end{definition}
\noindent
It is well known and easy to verify that for every graph $G$, $\vchrom(G) \leq \svchrom(G) \leq \chi(G)$.
The Lov\'{a}sz $\vartheta$-function, introduced in~\cite{Lovasz79}, is known to satisfy $\vartheta(G)=\svchrom(\overline{G})$ for every graph
$G$~\cite{KargerMS98} (see Section~\ref{sec:intro} for its original definition).

\subsection{Generalized Kneser Graphs}

Consider the family of generalized Kneser graphs defined below. In these graphs the vertices are subsets of some universe and the existence of an edge connecting two sets is decided according to their intersection size.

\begin{definition}\label{def:Kneser}
For integers $s \leq d$ and a set $T \subseteq \{0,1,\ldots,s-1\}$, the graph $K(d,s,T)$ is defined as follows: the vertices are all possible $s$-subsets of a universe $[d]$ (i.e., subsets of $[d]$ of size $s$), and two distinct sets $A,B$ are adjacent if $|A \cap B| \in T$.
\end{definition}

The following lemma provides a bound on the strict and non-strict vector chromatic numbers of certain generalized Kneser graphs.
Its proof can be found in~\cite[Section~9]{KargerMS98}, and we include it here for completeness.

\begin{lemma}[\cite{KargerMS98}]\label{lemma:chi_vK}
Let $t< s < d$ be integers satisfying $s^2 > dt$.
\begin{enumerate}
  \item\label{itm:chi_v} If $T = \{0,1,\ldots,t\}$ then $\chi_v(K(d,s,T)) \leq \frac{d(s-t)}{s^2-dt}$.
  \item\label{itm:chi_v_s} If $T = \{t\}$ then $\chi_v^{(s)}(K(d,s,T)) \leq \frac{d(s-t)}{s^2-dt}$.
\end{enumerate}
\end{lemma}

\begin{proof}
Associate every vertex $A$ of $K(d,s,T)$, representing an $s$-subset of $[d]$, with the vector $u_A \in \R^d$ defined by \[(u_A)_i = z \mbox{ ~if~ } i \in A \mbox{ ~and~ } (u_A)_i = - 1 \mbox{ ~if~ } i \notin A, \mbox{ ~for every~ } i \in [d],\]
where $z$ is a positive real number to be determined. Notice that $\|u_A\|^2 = s \cdot z^2 +d-s$.
Denote by $w_A \in \R^d$ the unit vector defined by $w_A = u_A / \|u_A\|$.

We start with Item~\ref{itm:chi_v}. Let $T = \{0,1,\ldots,t\}$.
Every two adjacent vertices $A$ and $B$ in $K(d,s,T)$ satisfy $|A \cap B| \leq t$, hence $|A \bigtriangleup B| \geq 2(s-t)$ and $|\overline{A \cup B}| \leq d-2s+t$. It follows that
\begin{eqnarray*}
\langle w_A, w_B \rangle &=& \frac{1}{s \cdot z^2 +d-s} \cdot \langle u_A, u_B \rangle \\
& = & \frac{1}{s \cdot z^2 +d-s} \cdot \Big ( |A \cap B| \cdot z^2 -|A \bigtriangleup B| \cdot z + |\overline{A \cup B}| \Big ) \\
& \leq & \frac{t \cdot z^2 -2(s-t) \cdot z + d-2s+t}{s \cdot z^2 +d-s}.
\end{eqnarray*}
A straightforward calculation shows that the minimum of the above expression is attained at $z = \frac{d}{s}-1$ and is equal to $-\frac{1}{\kappa-1}$ for $\kappa = \frac{d(s-t)}{s^2-dt}>1$. This completes the proof of Item~\ref{itm:chi_v}.
The proof of Item~\ref{itm:chi_v_s} is essentially identical. For adjacent vertices $A$ and $B$ in $K(d,s,T)$ where $T = \{t\}$ we have $|A \cap B| = t$, hence the above upper bound on $\langle w_A, w_B \rangle$ is tight, as needed for the strict vector chromatic number.
\end{proof}

\subsection{Linear Algebra Fact}

\begin{fact}\label{fact:rankp_R}
Let $p$ be a prime and let $M$ be an integer matrix. Then, the matrix $M' = M~(\mod~p)$ satisfies ${\rank}_p(M') \leq {\rank}_\R(M)$.
\end{fact}

\begin{proof}
It suffices to show that if some rows $v_1,\ldots, v_k$ of $M$ are linearly dependent over $\R$ then, considered modulo $p$, they are also linearly dependent over $\Fset_p$.
To see this, assume that there exist $a_1,\ldots,a_k \in \R$, at least one of which is nonzero, for which $\sum_{i=1}^{k}{a_i v_i}=0$. Since the $v_i$'s are integer vectors it can be assumed that $a_1,\ldots,a_k \in \Z$ and that $\gcd(a_1,\ldots,a_k)=1$. This implies that they are not all zeros modulo $p$. Therefore, the same coefficients, considered modulo $p$, provide a non-trivial combination of the corresponding rows of $M'$ with sum zero, and we are done.
\end{proof}

\section{Upper Bounds on Minrank via Multivariate Polynomials}\label{sec:poly}

In this section we prove upper bounds on the minrank parameter of graphs using linear spaces of multivariate polynomials.
We first introduce the notion of functional bi-representations of graphs.
\begin{definition}\label{def:functional}
Let $G=(V,E)$ be a directed graph and let $\Fset$ be a field.
A {\em functional bi-representation} of $G$ over $\Fset$ of dimension $R$ is an assignment of two functions $g_i, h_i : V \rightarrow \Fset$ to each $i \in [R]$ such that the function $f: V \times V \rightarrow \Fset$ defined by
\begin{eqnarray}\label{eq:f}
f(u,v) = \sum_{i=1}^{R}{g_i(u) h_i(v)}
\end{eqnarray}
satisfies
\begin{enumerate}
  \item $f(v,v) \neq 0$ for every $v \in V$, and
  \item $f(u,v) = 0$ for every distinct $u,v \in V$ such that $(u,v) \notin E$.
\end{enumerate}
Note that the definition is naturally extended to undirected graphs.
\end{definition}

Functional bi-representations can be used to provide an alternative definition for the minrank parameter, as stated below.

\begin{lemma}\label{lemma:minrank_def}
For every (directed) graph $G$ and a field $\Fset$, ${\minrank}_\Fset(G)$ is the smallest integer $R$ for which there exists a functional bi-representation of $G$ over $\Fset$ of dimension $R$.
\end{lemma}
\noindent
Observe that Lemma~\ref{lemma:minrank_def} follows directly from Definition~\ref{def:minrank} and the linear algebra fact that the rank of a matrix $M \in \Fset^{N \times N}$ is the smallest $R$ for which $M = A \cdot B$ for two matrices $A \in \Fset^{N \times R}$ and $B \in \Fset^{R \times N}$, where the functions $g_i$ and $h_i$ in Definition~\ref{def:functional} correspond to the columns and the rows of such $A$ and $B$ respectively.
A similar definition of the minrank parameter was previously used by Peeters in~\cite{Peeters96} (see also~\cite{ChlamtacH14}), where the role of the functions in Definition~\ref{def:functional} was taken by vectors.

\subsection{Generalized Kneser Graphs}

We prove now upper bounds on the minrank of generalized Kneser graphs $K(d,s,T)$ (recall Definition~\ref{def:Kneser}).
The proofs borrow ideas from~\cite{AlonBS91} and~\cite{AlonUnion98}.
We start with an upper bound on the minrank of $K(d,s,T)$ over $\Fset_p$ for all sufficiently large primes $p$. Note that a slightly improved bound is given in Section~\ref{sec:incidence} (see Proposition~\ref{prop:minrank_Kneser}).

\begin{proposition}\label{prop:minrank_K(d,s,T)}
For every integers $t \leq s \leq d$, a set $T \subseteq \{0,1,\ldots,s-1\}$ of size $|T|=t$, and a prime $p > s$, \[{\minrank}_p( K(d,s,T) ) \leq \sum_{i=0}^{s-t}{\binom d {i}}.\]
\end{proposition}

\begin{proof}
Let $p > s$ be a prime, and let $f: \{0,1\}^{d} \times \{0,1\}^{d} \rightarrow \Fset_p$ be the function defined by
\[ f(x,y) = \prod_{j \in \{0,1,\ldots,s-1\} \setminus T}{\Big ( \sum_{i=1}^{d}{x_i y_i} -j \Big )}~~~(\mod~p)\]
for every $x,y \in \{0,1\}^d$.
Expanding $f$ as a linear combination of monomials, the relation $z^2=z$ for $z \in \{0,1\}$ implies that one can reduce to $1$ the exponent of each variable occurring in a monomial. It follows that $f$ can be represented as a multilinear polynomial in the $2d$ variables of $x$ and $y$.
By combining terms involving the same monomial in the variables of $x$, one can write $f$ as in~\eqref{eq:f}
for an integer $R$ and functions $g_i,h_i : \{0,1\}^d \rightarrow \Fset_p$ such that the $g_i$'s are distinct multilinear monomials of total degree at most $s-t$ in $d$ variables. It follows that $R \leq \sum_{i=0}^{s-t}{{\binom d {i}}}$.

Now, denote by $V$ the vertex set of the graph $K(d,s,T)$ and identify each vertex $X \in V$ with an indicator vector $c_X \in \{0,1\}^d$ in the natural way.
We observe that the functions $g_i$ and $h_i$ restricted to $V$ form a functional bi-representation of $K(d,s,T)$ over $\Fset_p$.
Indeed, for every two vertices $A,B \in V$ we have
$f(c_A,c_B) = \prod_{j \in \{0,1,\ldots,s-1\} \setminus T}{( |A \cap B| -j)}~(\mod~p)$.
If $A$ and $B$ are distinct and non-adjacent in $K(d,s,T)$ then $|A \cap B| \in \{0,1,\ldots,s-1\} \setminus T$, and thus $f(c_A,c_B) = 0$.
On the other hand, every vertex $A$ satisfies $|A| = s$ and thus, using the assumption $p>s$, $f(c_A,c_A) \neq 0$.
By Lemma~\ref{lemma:minrank_def} it follows that ${\minrank}_p( K(d,s,T) ) \leq R$, and we are done.
\end{proof}

We next consider a special case of the generalized Kneser graphs corresponding to one intersection size.

\begin{proposition}\label{prop:Gp}
For every prime $p$ and integer $d \geq 2p-1$, ${\minrank}_p(K(d,2p-1,\{p-1\})) \leq \sum_{i=0}^{p-1}{\binom {d} {i}}$.
\end{proposition}

\begin{proof}
For a prime $p$ and an integer $d \geq 2p-1$, consider the graph $G = K(d,2p-1, \{p-1\})$.
Let $f: \{0,1\}^{d} \times \{0,1\}^{d} \rightarrow \Fset_p$ be the function defined by
\[ f(x,y) = \prod_{j=0}^{p-2}{\Big ( \sum_{i=1}^{d}{x_i y_i}-j \Big )}~~(\mod~p)\]
for every $x,y \in \{0,1\}^d$.
By a repeated use of the relation $z^2=z$ for $z \in \{0,1\}$, the function $f$ can be represented as a multilinear polynomial in the $2d$ variables of $x$ and $y$. By combining terms involving the same monomial in the variables of $x$, it follows that one can write $f$ as in~\eqref{eq:f}
for an integer $R$ and functions $g_i,h_i : \{0,1\}^d \rightarrow \Fset_p$ such that the $g_i$'s are distinct multilinear monomials of total degree at most $p-1$ in $d$ variables. It thus follows that $R \leq \sum_{i=0}^{p-1}{\binom {d} {i}}$.

Now, denote by $V$ the vertex set of $G$ and, as before, identify each vertex $X \in V$ with an indicator vector $c_X \in \{0,1\}^d$ in the natural way.
We observe that the functions $g_i$ and $h_i$ restricted to $V$ form a functional bi-representation of $G$ over $\Fset_p$.
Indeed, for every two vertices $A,B \in V$ we have
$f(c_A,c_B) = 0$ if and only if $| A \cap B| \neq p-1~(\mod~p)$.
If $A$ and $B$ are distinct non-adjacent vertices in $G$ then $|A \cap B| \neq p-1$, so since $|A|=|B|=2p-1$ it follows that $|A \cap B| \neq p-1~(\mod~p)$ as well, thus $f(c_A,c_B) = 0$.
On the other hand, every $A \in V$ satisfies $|A| = 2p-1$, so $|A| = p-1~(\mod~p)$, and thus $f(c_A,c_A) \neq 0$.
By Lemma~\ref{lemma:minrank_def} it follows that ${\minrank}_p( G ) \leq R$, and we are done.
\end{proof}

\subsection{Orthogonality versus Non-orthogonality}\label{subsec:ortho}

For a prime $p$ and an integer $d \geq 1$, let $G_1(d,p)$ be the graph whose vertex set $V$ consists of the non-self-orthogonal vectors of $\Fset_p^d$, such that two distinct vertices are adjacent if they are not orthogonal over $\Fset_p$.
The minrank of $G_1(d,p)$ over $\Fset_p$ is equal to $d$.
For the lower bound, observe that $G_1(d,p)$ contains an independent set of size $d$.
For the upper bound, consider the $|V| \times d$ matrix $M$ over $\Fset_p$ in which the row indexed by a vertex $v \in V$ is $v$, and notice that the matrix $M \cdot M^T$ represents $G_1(d,p)$ and that its rank is at most $d$.
Note that variants of the graph $G_1(d,p)$ were found useful in the study of the minrank parameter (see, e.g.,~\cite{Peeters96} and~\cite[Section ~4.1]{BlasiakKL13}).

It is natural to consider a variant of $G_1(d,p)$ in which the vertices are replaced by the self-orthogonal vectors of $\Fset_p^d$ and the edges are defined in the same way. Namely, let $G_2(d,p)$ be the graph whose vertex set consists of the self-orthogonal vectors of $\Fset_p^d$, such that two distinct vertices are adjacent if they are not orthogonal over $\Fset_p$. We prove below that in contrast to $G_1(d,p)$ the minrank of $G_2(d,p)$ over $\Fset_p$, for a fixed $p$, grows exponentially in $d$.
To this end, we prove the following upper bound on the minrank of its complement.
The proof is inspired by an idea used in the context of matching vector codes by Dvir, Gopalan, and Yekhanin~\cite{DvirGY11}.

\begin{proposition}\label{prop:G_2}
For every prime $p$ and an integer $d \geq 1$, ${\minrank}_p(\overline{G_2(d,p)}) \leq {\binom {d+p-2} {p-1}} +1$.
\end{proposition}

\begin{proof}
Let $V$ be the vertex set of $G_2(d,p)$, i.e., the set of self-orthogonal vectors of $\Fset_p^d$.
Let $f: V \times V \rightarrow \Fset_p$ be the function defined by
\[ f(x,y) = 1 -  \Big (\sum_{i=1}^{d}{x_i y_i} \Big )^{p-1}~~(\mod~p)\]
for every $x,y \in V$.
Expanding $f$ as a linear combination of monomials and combining terms involving the same monomial in the variables of $x$, it follows that $f$ can be written as in~\eqref{eq:f}
for an integer $R$ and functions $g_i, h_i : V \rightarrow \Fset_p$, where $g_1=1$ and the $g_i$'s for $i \geq 2$ are distinct monomials of degree exactly $p-1$ in $d$ variables. This yields that $R \leq {\binom {d+p-2} {p-1}} +1$.

Now, let us show that the functions $g_i$ and $h_i$ form a functional bi-representation of $\overline{G_2(d,p)}$ over $\Fset_p$. Indeed, by Fermat's little Theorem, $f(u,v) \neq 0$ if and only if $u$ and $v$ are orthogonal. If $u$ and $v$ are distinct non-adjacent vertices in $\overline{G_2(d,p)}$ then they are not orthogonal, thus $f(u,v)=0$. On the other hand, by the self-orthogonality of the vectors in $V$, we have $f(v,v) \neq 0$ for every $v \in V$.
By Lemma~\ref{lemma:minrank_def} it follows that ${\minrank}_p( \overline{G_2(d,p)} ) \leq R$, and we are done.
\end{proof}

Combining Proposition~\ref{prop:G_2} with Lemma~\ref{lemma:minrank_comp} implies the following.

\begin{corollary}
For every prime $p$ and an integer $d \geq 1$,
\[{\minrank}_p(G_2(d,p)) \geq \frac{n}{{\binom {d+p-2} {p-1}} +1},\]
where $n$ stands for the number of vertices in $G_2(d,p)$. In particular, using $n \geq p^{d-p+1}$, for a fixed prime $p$,
\[{\minrank}_p(G_2(d,p)) \geq p^{(1-o(1)) \cdot d}.\]
\end{corollary}

\subsection{A Directed Example}\label{subsec:directed}

We end this section with a quick application of multivariate polynomials to the minrank of a directed graph.
The proof employs an idea of Blokhuis~\cite{Blokhuis93} used in the study of the Sperner capacity of the cyclic triangle.

\begin{proposition}
For an integer $d \geq 1$, let $G = (V,E)$ be the directed graph on $V = \{0,1,2\}^d$ where for every distinct $u,v \in V$, $(u,v) \in E$ if $u-v \in \{0,2\}^d~(\mod~3)$. Then, ${\minrank}_3(G) = 2^d$.
\end{proposition}

\begin{proof}
We start with the upper bound.
Let $f: V \times V \rightarrow \Fset_3$ be the function defined by
\[ f(x,y) = \prod_{i=1}^{d}{(x_i-y_i-1)} ~~(\mod~3) \]
for every $x,y \in V$.
Expanding $f$ as a linear combination of monomials and combining terms involving the same monomial in the variables of $x$, it follows that one can write $f$ as in~\eqref{eq:f}
for an integer $R$ and functions $g_i, h_i : V \rightarrow \Fset_p$, where the $g_i$'s are distinct multilinear monomials in $d$ variables. This yields that $R \leq 2^d$.

We now show that the functions $g_i$ and $h_i$ form a functional bi-representation of $G$ over $\Fset_3$.
Indeed, for every distinct vertices $u,v \in V$, if $(u,v) \notin E$ then $u_i - v_i = 1~(\mod~3)$ for some $i$, and thus $f(u,v)=0$.
On the other hand, for every $v \in V$, $f(v,v) = (-1)^d \neq 0$.
By Lemma~\ref{lemma:minrank_def} it follows that ${\minrank}_3( G ) \leq R$, as desired.

For the lower bound, observe that the subgraph of $G$ induced by the set of vertices $\{0,1\}^d$ is acyclic.
It was shown in~\cite{BBJK06} that the maximum size of an induced acyclic subgraph forms a lower bound on the minrank parameter, thus the proof is completed.
\end{proof}

\section{Upper Bounds on Minrank via Inclusion Matrices}\label{sec:incidence}

In this section we prove upper bounds on the minrank parameter of graphs using the method of higher incidence matrices, more specifically -- the class of inclusion matrices. The proofs employ ideas from~\cite{FranklW81}. We start with a few notations and facts following~\cite[Chapter~7]{BabaiF92}.

\paragraph{Binomial coefficient polynomials. }
For an integer $k \geq 0$, the binomial coefficient ${\binom x k}$ is a polynomial of degree $k$ over $\R$ defined by
\[ {\binom x k} = \frac{1}{k!} \cdot x(x-1) \cdots (x-k+1).\]
We say that a polynomial is {\em integer-valued} if it takes integer values on integers. We need the following fact (see, e.g.,~\cite[Exercise~7.3.3]{BabaiF92}).

\begin{fact}\label{fact:integer-valued}
For every $k \geq 0$, the integer-valued polynomials of degree at most $k$ are precisely all the integer linear combinations of the polynomials ${\binom x 0}, {\binom x 1}, \ldots, {\binom x k}$.
\end{fact}

\paragraph{Inclusion matrices.}
For integers $d \geq s \geq k \geq 0$, let $N^{(d)}(s,k)$ denote the ${\binom d s} \times {\binom d k}$ binary matrix, whose rows and columns are indexed by all $s$-subsets and $k$-subsets of $[d]$ respectively, defined by
\[ (N^{(d)}(s,k))_{A,B} = 1 \mbox{~~if and only if~~} B \subseteq A\]
for every $s$-subset $A$ and $k$-subset $B$ of $[d]$.
In addition, let $M^{(d)}(s,k)$ denote the ${\binom d s} \times {\binom d s}$ integer matrix defined by
\begin{eqnarray}\label{eq:Mk_basic}
M^{(d)}(s,k) = N^{(d)}(s,k) \cdot N^{(d)}(s,k)^T.
\end{eqnarray}
Notice that the entry of $M^{(d)}(s,k)$ indexed by $(A,B)$, where $A,B$ are $s$-subsets of $[d]$, precisely counts the $k$-subsets $X$ of $[d]$ that satisfy $X \subseteq A \cap B$. Hence, for every $s$-subsets $A$ and $B$ of $[d]$,
\begin{eqnarray}\label{eq:Mk}
(M^{(d)}(s,k))_{A,B} = {\binom {|A \cap B|} k}.
\end{eqnarray}

\begin{lemma}\label{lemma:rank_incidence}
For every $d \geq s \geq \ell \geq 0$ and $a_0,\ldots,a_\ell \in \R$, the matrix $M = \sum_{k=0}^{\ell}{a_k \cdot M^{(d)}(s,k)}$ satisfies ${\rank}_\R(M) \leq {\binom d \ell}$.
\end{lemma}

\begin{proof}
We first claim that for every $0 \leq k \leq \ell$, every column of $M^{(d)}(s,k)$ is a linear combination of the columns of $N^{(d)}(s,\ell)$.
For $k=\ell$ this follows immediately from~\eqref{eq:Mk_basic}.
To see this for $0 \leq k < \ell$, consider the ${\binom d s} \times {\binom d k}$ matrix $N^{(d)}(s,\ell) \cdot N^{(d)}(\ell,k)$, and observe that the entry indexed by $(A,B)$ in this matrix, where $A,B \subseteq [d]$, $|A|=s$, $|B|=k$, counts the $\ell$-subsets $X$ of $[d]$ that satisfy $B \subseteq X \subseteq A$. If $B \subseteq A$ then the number of these subsets is ${\binom {s-k} {\ell-k}}$ and otherwise it is $0$. It follows that
\[ N^{(d)}(s,\ell) \cdot N^{(d)}(\ell,k) =  {\binom {s-k} {\ell-k}} \cdot N^{(d)}(s,k).\]
Hence, every column of $N^{(d)}(s,k)$ is a linear combination of the columns of $N^{(d)}(s,\ell)$, and by~\eqref{eq:Mk_basic}, the same holds for the columns of $M^{(d)}(s,k)$ where $0 \leq k \leq \ell$. As the matrix $M$ is a linear combination of these matrices, it follows that its columns lie in the space spanned by the columns of $N^{(d)}(s,\ell)$ whose dimension is at most ${\binom d \ell}$. This yields the required bound on the rank of $M$.
\end{proof}

\subsection{Generalized Kneser Graphs}

As our first application of the method of inclusion matrices, we improve the bound given in Proposition~\ref{prop:minrank_K(d,s,T)} for the generalized Kneser graphs $K(d,s,T)$ (recall Definition~\ref{def:Kneser}). We note, though, that this improvement is not essential to our application in Section~\ref{sec:separation}.

\begin{proposition}\label{prop:minrank_Kneser}
For every integers $t \leq s \leq d$, a set $T \subseteq \{0,1,\ldots,s-1\}$ of size $|T|=t$, and a prime $p > s$, \[{\minrank}_p( K(d,s,T) ) \leq {\binom d {s-t}}.\]
\end{proposition}

\begin{proof}
Consider the polynomial $m \in \R [x]$ defined by
\[m(x) = \prod_{j \in \{0,1,\ldots,s-1\} \setminus T}{(x-j)}.\]
Notice that $m$ is an integer-valued polynomial of degree $s-t$. By Fact~\ref{fact:integer-valued}, one can write
\begin{eqnarray}\label{eq:q}
m(x) = \sum_{k=0}^{s-t}{a_k \cdot {\binom x k}}
\end{eqnarray}
for integer coefficients $a_0, \ldots, a_{s-t}$.
Using the notation in~\eqref{eq:Mk_basic}, define the ${\binom d s} \times {\binom d s}$ integer matrix
\[M = \sum_{k=0}^{s-t}{a_k \cdot M^{(d)}(s,k)},\]
and let $M' = M~(\mod~p)$ for a prime $p > s$.
We claim that $M'$ represents the graph $K(d,s,T)$ over $\Fset_p$.
Indeed, using~\eqref{eq:Mk} and~\eqref{eq:q}, for every two $s$-subsets $A$ and $B$ of $[d]$ we have
\[M_{A,B} = \sum_{k=0}^{s-t}{a_k \cdot (M^{(d)}(s,k))_{A,B}} = \sum_{k=0}^{s-t}{a_k \cdot {\binom {|A \cap B|} k}} = m(|A \cap B|).\]
Every two distinct vertices $A,B$ non-adjacent in $K(d,s,T)$ satisfy $|A \cap B| \in \{0,1,\ldots,s-1\} \setminus T$, and thus $M_{A,B} = m(|A \cap B|) = 0$, and in particular $M'_{A,B} = 0$. On the other hand, every vertex $A$ satisfies $|A| = s$, and thus $M_{A,A} = m(s)$, so using the assumption $p>s$ it follows that $M'_{A,A} \neq 0$.
Finally, we obtain that
\[{\minrank}_p( K(d,s,T) ) \leq {\rank}_p(M') \leq {\rank}_\R (M) \leq {\binom d {s-t}},\]
where the second and third inequalities follow from Fact~\ref{fact:rankp_R} and Lemma~\ref{lemma:rank_incidence} respectively, and we are done.
\end{proof}

We next consider a modular variant of the generalized Kneser graphs, defined as follows.

\begin{definition}\label{def:Kneser=q}
For integers $t \leq s \leq d$ and $q$, the graph $K_q(d,s,t)$ is defined as $K(d,s,T)$ where $T = \{ i \in \{0,1, \ldots,s-1\} \mid i = t ~(\mod~q)\}$.
\end{definition}

The following proposition provides an upper bound on the minrank over $\Fset_p$ of the graph $K_q(d,s,t)$, where $p$ is a prime and $q$ is a power of $p$.
This bound is crucial for the construction given in Theorem~\ref{thm:Intro1}, which separates for every {\em fixed} prime $p$ the $\vartheta$-function of a graph from its minrank over $\Fset_p$.

\begin{proposition}\label{prop:Kneser=q}
For every prime $p$, a prime power $q = p^\ell$, and integers $t \leq s \leq d$ such that $q \leq s+1$ and  $s = t ~(\mod~q)$,
\[{\minrank}_p(K_q(d,s,t)) \leq {\binom d {q-1}}.\]
\end{proposition}

\begin{proof}
Let $p$ be a prime and let $q = p^\ell$ be a prime power.
Consider the polynomial $m \in \R[x]$ defined by $m(x) = {\binom {x-t-1} {q-1}}$.
By Fact~\ref{fact:integer-valued}, $m$ is an integer-valued polynomial of degree $q-1$, which can be written as
\begin{eqnarray}\label{eq:m}
m(x) = \sum_{k=0}^{q-1}{a_k \cdot {\binom x k}}
\end{eqnarray}
for integer coefficients $a_0, \ldots, a_{q-1}$.
Using the notation in~\eqref{eq:Mk_basic}, define the ${\binom d s} \times {\binom d s}$ integer matrix
\[M = \sum_{k=0}^{q-1}{a_k \cdot M^{(d)}(s,k)},\]
and let $M' = M~(\mod~p)$.
We claim that $M'$ represents the graph $K_q(d,s,T)$ over $\Fset_p$.
Indeed, using~\eqref{eq:Mk} and~\eqref{eq:m}, for every two $s$-subsets $A$ and $B$ of $[d]$ we have
\[M_{A,B} = \sum_{k=0}^{q-1}{a_k \cdot (M^{(d)}(s,k))_{A,B}} = \sum_{k=0}^{q-1}{a_k \cdot {\binom {|A \cap B|} k}} = m(|A \cap B|).\]
To complete the argument, we need the following fact (see, e.g.,~\cite[Proposition~5.31]{BabaiF92}).
\begin{fact}\label{fact:binom}
For every prime $p$, a prime power $q = p^\ell$ and an integer $r$, $p$ divides ${\binom {r-1} {q-1}}$ if and only if $q$ does not divide $r$.
\end{fact}
\noindent
By Fact~\ref{fact:binom} and the definition of $m$, $M'_{A,B} = 0$ if and only if $|A \cap B| \neq t ~(\mod~q)$.
Every two distinct vertices $A,B$ non-adjacent in $K_q(d,s,t)$ satisfy $|A \cap B| \neq t ~(\mod~q)$, and thus $M'_{A,B} = 0$.
On the other hand, every vertex $A$ satisfies $|A| = s$, so using the assumption $s = t ~(\mod~q)$, it follows that $M'_{A,A} \neq 0$.
Finally, we obtain that
\[{\minrank}_p( K_q(d,s,t) ) \leq {\rank}_p(M') \leq {\rank}_\R (M) \leq {\binom d {q-1}},\]
where the second inequality follows from Fact~\ref{fact:rankp_R} and the third follows from Lemma~\ref{lemma:rank_incidence} using $q \leq s+1$, so we are done.
\end{proof}

\section{Separations between Minrank and Other Graph Parameters}\label{sec:separation}

In this section we prove Theorems~\ref{thm:Intro1},~\ref{thm:Intro2}, and~\ref{thm:Intro3}.
We start with the proof of Theorem~\ref{thm:Intro3}, which claims the existence of $n$-vertex graphs whose minrank, over any sufficiently large prime order field, is polynomial in $n$ while their complement is vector $3$-colorable.
The proof is based on instances of the generalized Kneser graphs, in which pairs of sets are adjacent if their intersection size is small. Such graphs were used in~\cite{KargerMS98} to provide a similar separation between the vector chromatic number and the chromatic number.

\begin{proof}[ of Theorem~\ref{thm:Intro3}]
For a sufficiently large integer $t$ define $d = 8t$, $s = 4t$, and $T = \{0,1,\ldots,t\}$.
Let $G$ be the complement of the graph $K(d,s,T)$ given in Definition~\ref{def:Kneser}, and note that the number $n$ of its vertices satisfies $n = {\binom {8t} {4t}} = 2^{(1-o(1))d}$.
By Item~\ref{itm:chi_v} of Lemma~\ref{lemma:chi_vK}, $\chi_v(\overline{G}) \leq \frac{d(s-t)}{s^2-dt} = 3$.
Apply Proposition~\ref{prop:minrank_Kneser} to get that for any prime $p > s = \Theta( \log n)$, we have
\[{\minrank}_p(\overline{G}) = {\minrank}_p(K(d,s,T)) \leq {\binom {d} {s-|T|}} = {\binom {8t} {3t-1}}.\]
By Lemma~\ref{lemma:minrank_comp}, this implies that
\[ {\minrank}_p(G) \geq \frac{n}{{\binom {8t} {3t-1}}} \geq n^{1-H(3/8)-o(1)} \geq n^{0.0455}, \]
where $H$ stands for the binary entropy function. This completes the proof.
\end{proof}

We turn to prove Theorem~\ref{thm:Intro1}, which claims that for every fixed prime $p$ there exist $n$-vertex graphs with constant $\vartheta$-function and minrank over $\Fset_p$ polynomial in $n$. Here we use the modular variant of the generalized Kneser graphs considered in Proposition~\ref{prop:Kneser=q} (recall Definition~\ref{def:Kneser=q}). For $p=2$, our graphs are related to a construction of~\cite{AlonP91}.

\begin{proof}[ of Theorem~\ref{thm:Intro1}]
We first prove the theorem for $p=2$.
For a sufficiently large integer $\ell$, let $d = 2^\ell$, $q = \frac{d}{4}$, $t = \frac{d}{8}$, and $s = t+q = \frac{3}{8} \cdot d$.
Let $G$ be the complement of the graph $K_q(d,s,t)$ given in Definition~\ref{def:Kneser=q}, and let $n = {\binom d s}$ denote the number of its vertices.
Recall that two distinct vertices $A$ and $B$, representing $s$-subsets of $[d]$, are adjacent in $K_q(d,s,t)$ if and only if $|A \cap B| = t~(\mod~q)$.
By $0 \leq t<q$ and $s = t+q$, this condition is equivalent for distinct $A$ and $B$ to $|A \cap B| =t$, so in our setting $K_q(d,s,t) = K(d,s,\{t\})$.
Recalling that $\vartheta(G)$ is equal to the strict vector chromatic number of $\overline{G}$, by Item~\ref{itm:chi_v_s} of Lemma~\ref{lemma:chi_vK} we obtain that
\[ \vartheta(G) = \chi_v^{(s)}(K(d,s,\{t\})) \leq \frac{d(s-t)}{s^2-dt} = \frac{1/4}{(3/8)^2-1/8} = 16.\]
Now, as $q$ is a power of $2$ and $s = t ~(\mod~q)$, we can apply Proposition~\ref{prop:Kneser=q} to obtain that
\[{\minrank}_2(\overline{G}) = {\minrank}_2(K_q(d,s,t)) \leq {\binom d {q-1}} \leq 2^{H ( 1/4 ) \cdot d},\]
where $H$ stands for the binary entropy function.
By Lemma~\ref{lemma:minrank_comp}, it follows that
\[ {\minrank}_2(G) \geq \frac{n}{2^{H (1/4) \cdot d}} \geq n^{1-\frac{H(1/4)}{H(3/8)}-o(1)} \geq n^{0.1499},\]
and we are done.

The proof for a general prime $p \geq 3$ is similar. Details follow.
For a sufficiently large integer $\ell$, let $d = p^\ell$, $q = \frac{d}{p}$, $t = \frac{d}{p^2}$, and $s = t+q = \frac{p+1}{p^2} \cdot d$.
As in the case of $p=2$, let $G$ be the complement of the graph $K_q(d,s,t) = K(d,s,\{t\})$, and let $n = {\binom d s}$ denote the number of its vertices.
By Item~\ref{itm:chi_v_s} of Lemma~\ref{lemma:chi_vK} we obtain that
\[ \vartheta(G) = \chi_v^{(s)}(K(d,s,\{t\})) \leq \frac{d(s-t)}{s^2-dt} = \frac{\frac{1}{p}}{\frac{(p+1)^2}{p^4}-\frac{1}{p^2}} = \frac{p^3}{2p+1}.\]
In particular, $\vartheta(G) \leq c$ for some $c = c(p)$.
As $q$ is a power of the prime $p$ and $s = t ~(\mod~q)$, we can apply Proposition~\ref{prop:Kneser=q} to obtain that
\[{\minrank}_p(\overline{G}) = {\minrank}_p(K_q(d,s,t)) \leq {\binom d {q-1}} \leq 2^{H ( 1/p ) \cdot d}.\]
By Lemma~\ref{lemma:minrank_comp}, using $p \geq 3$ and the monotonicity of $H$ in $[0,0.5]$, it follows that
\[ {\minrank}_p(G) \geq \frac{n}{2^{H (1/p) \cdot d}} \geq n^{1-\frac{H(1/p)}{H( (p+1)/p^2)}-o(1)} \geq n^{\delta},\]
for some $\delta = \delta(p)>0$, completing the proof.
\end{proof}

Finally, we prove the following theorem that confirms Theorem~\ref{thm:Intro2}.
Here we use the generalized Kneser graphs considered in Proposition~\ref{prop:Gp}.

\begin{theorem}\label{thm:ThetaMinrankp}
For any $\delta < 0.1887$ there exists $c = c(\delta)$ such that for infinitely many integers $n$ there exists an $n$-vertex graph $G$ such that $\vartheta(G) \leq c$ and $\minrank_p(G) \geq n^{\delta}$ for some $p = \Theta(\log n)$.
\end{theorem}

\begin{proof}
For a sufficiently large prime $p$, let $\epsilon \in (0,2)$ be a real number such that $d = (4-\epsilon) \cdot p$ is an integer, and let $s = 2p-1$ and $t=p-1$.
Let $G$ be the complement of the graph $K(d,s,\{t\})$.
Since $s^2 > dt$ we can apply Item~\ref{itm:chi_v_s} of Lemma~\ref{lemma:chi_vK} to obtain that
\begin{eqnarray*}
 \vartheta(G) &=& \chi_v^{(s)}(\overline{G}) = \chi_v^{(s)}(K(d,s,\{t\})) \leq \frac{d(s-t)}{s^2-dt} \\
 &=& \frac{(4-\epsilon) p^2}{(2p-1)^2-(4-\epsilon)p(p-1)} = \frac{(4-\epsilon)p^2}{\eps p^2 - \eps p +1} \leq \frac{(4-\eps)p^2}{\eps p^2/2} \leq \frac{2(4-\epsilon)}{\epsilon},
 \end{eqnarray*}
where in the second inequality we have used the assumption that $p$ is sufficiently large.
Now, by Proposition~\ref{prop:Gp} it follows that
\[{\minrank}_p(\overline{G})  = {\minrank}_p(K(d,s,\{t\})) \leq \sum_{i=0}^{p-1}{{\binom {d} {i}}} \leq 2^{H(1/(4-\eps)) \cdot d}.\]
Let $n$ denote the number of vertices in $G$, and notice that $n = {\binom {d} {2p-1}} = {\binom {d} {d-2p+1}}$.
Applying Lemma~\ref{lemma:minrank_comp}, we get that
\[ {\minrank}_p(G) \geq \frac{n}{2^{H(1/(4-\eps)) \cdot d}} \geq n^{1-\frac{H(1/(4-\eps))}{H((2-\epsilon)/(4-\eps))}-o(1)},\]
where $p = \Theta(d) = \Theta(\log n)$.

Finally, notice that for every $\delta < 1-H(1/4) \approx 0.1887$ one can choose a sufficiently small $\eps >0$ for which the above construction gives an $n$-vertex graph $G$ with $\vartheta(G) \leq c$ and ${\minrank}_p(G) \geq n^\delta$, where $c$ depends only on $\delta$ and $p = \Theta(\log n)$.
\end{proof}

\bibliographystyle{abbrv}
\bibliography{theta}

\end{document}